\documentclass[11pt]{article}
\usepackage[letterpaper, left=1.25in, right=1.25in, top=1.25in, bottom=1.25in]{geometry}
\usepackage[rm,bf]{titlesec}
\usepackage{enumitem}
\usepackage{array}

\usepackage{amsmath}
\usepackage{amsthm}
\usepackage{multirow}
\usepackage{booktabs}

\usepackage{graphicx}
\usepackage{amssymb}
\usepackage{amsfonts}
\usepackage{amsmath,amsthm}
\usepackage[footnotesize,bf]{caption}
\usepackage{natbib}
\usepackage{setspace}
\usepackage[dvipsnames,usenames]{color}
\usepackage{pstricks,pstricks-add}
\usepackage{url}
\usepackage[usenames]{color}
\usepackage[breaklinks, colorlinks, citecolor=blue, linkcolor=blue, urlcolor=blue]{hyperref}
\usepackage{lscape}
\usepackage{pst-solides3d}

\newtheorem{theorem}{Theorem}

\newtheorem{theorem*}[theorem]{Theorem*}

\newtheorem{definition}{Definition}
\newtheorem{definition*}[definition]{Definition*}

\newtheorem{example}{Example}
\newtheorem{example*}[example]{Example*}

\newtheorem{lemma}{Lemma}
\newtheorem{lemma*}[lemma]{Lemma*}

\newtheorem{proposition}{Proposition}
\newtheorem{proposition*}[proposition]{Proposition*}

\newcommand{\Ncal}{\mathcal{N}}
\newcommand{\Mcal}{\mathcal{M}}
\newcommand{\Bcal}{\mathcal{B}}

\newcommand{\N}{\mathbb{N}}
\newcommand{\R}{\mathbb{R}}
\newcommand{\U}{\mathbb{U}}
\newcommand{\M}{\mathbb{M}}
\newcommand{\Ob}{\mathbb{O}}

\newcommand{\Phen}{\mathbb{P}}

\newcommand{\Nash}{\textrm{N}}
\newcommand{\BNE}{\textrm{BNE}}

\newpsobject{grilla}{psgrid}{subgriddiv=1,griddots=10,gridlabels=6pt}

\begin{document}
\bibliographystyle{elsart-harv}
\title{The paradox of monotone structural QRE\footnote{Thanks to the audience at UT Austin Theory seminar for useful comments. All errors are our own.}}
\date{\today}
%\date{October 7, 2016}

\author{Rodrigo A. Velez\thanks{
\href{mailto:rvelezca@tamu.edu}{rvelezca@tamu.edu}; \href{https://sites.google.com/site/rodrigoavelezswebpage/home}{https://sites.google.com/site/rodrigoavelezswebpage/home}}\ \ and Alexander L. Brown\thanks{
 \href{mailto:alexbrown@tamu.edu}{alexbrown@tamu.edu}; \href{http://people.tamu.edu/\%7Ealexbrown}{http://people.tamu.edu/$\sim$alexbrown}} \\\small{\textit{Department of
Economics, Texas A\&M University, College Station, TX 77843}}}
\maketitle

\begin{abstract}
%\begin{singlespace}
\cite{mckelvey:95geb}'s monotone structural Quantal Response Equilibrium theory may be misspecified for the study of monotone behavior.
%\end{singlespace}
\medskip
\begin{singlespace}

\medskip

\medskip

\end{singlespace}
\end{abstract}

%\newpage
\section{Introduction}\label{Sec:intro}

The structural Quantal Response Equilibrium (sQRE) of \cite{mckelvey:95geb} is one of the most popular theories for the analysis of data from laboratory experiments \citep{Goeree-Holt-Palfrey-2016-Book}. This theory differs from the basic Nash equilibrium theory in which it assumes that there are \textit{separable} unobservable perturbations to agents' payoffs, which induce agents to behave as if they were noisy best responders.

sQRE theory has been subject to criticism. If its unobservables are unrestricted in the domain in which they are defined, for each possible distribution on observables there are unobservables that generate it \citep{Haile-et-al-2008}. As a response to this, \citet{Goeree-Holt-Palfrey-2005-EE} proposed two solutions. The first is to discipline sQRE with an a priori property of behavior for which there is empirical support, \textit{payoff monotonicity}. This property requires frequencies of play of each agent to be ordinally equivalent with expected utility. The second is to replace the structure of unobservables in this theory with a reduced form noisy best response, which also guarantees payoff monotonicity, the regular QRE (rQRE) theory.

\citet{Goeree-Holt-Palfrey-2005-EE} argue that the separability of perturbations in sQRE induces restrictions on behavior that may not be plausible when comparing behavior across games. Because of this they advocate rQRE in these situations. They do not have conclusive recommendations for the analysis of data from games in which unobservables are plausible to be comparable. For instance, for different experimental sessions with the same payoff function when subjects come from the same population. Indeed, they pose as an open question whether in this situation payoff monotone sQRE and rQRE have the same empirical content \citep[Sec. 6.1.,][]{Goeree-Holt-Palfrey-2005-EE}.

\citet{Haile-et-al-2008}'s criticism of sQRE and  \citet{Goeree-Holt-Palfrey-2005-EE}'s response to it epitomize the conflict between falsifiability and specification of a theory.  We submit that when facing this trade off, one can gain a better understanding of the problem by comparing the empirical content of a falsifiable theory with the ``bare'' empirical content of the a priori restrictions used to make the model falsifiable. That is, one should check whether all behavior satisfying these a priori restrictions can be generated by the falsifiable model. If this is so, one is sure that as long as these restrictions are supported by data, the model is well specified (under the implicit hypothesis that the structure of the unobservables is the correct one). If this is not so, either one should be able to confirm that empirical data can be characterized by stronger properties, or the model may be misspecified and its structure of unobservables needs to be revised.

This paper advances this study for sQRE when restricted by payoff monotonicity. We show that whenever at least an agent has at least three actions available, there can be payoff monotone distributions of observables that \emph{cannot} be generated by any  sQRE model consistent with payoff monotonicity (Theorem~\ref{Thm:Paradox}). This paradoxical situation is actually not resolved if one also disciplines this theory with proximity to Nash behavior. That is, for a given action space in which at least an agent has at least three actions available, one can always construct a payoff matrix for which the union of the range of all payoff monotone sQRE models excludes an open ball centered in a Nash equilibrium that is the limit of payoff monotone behavior (Theorem~\ref{Thm:Paradox2}).

All in all, our results alert researchers in experimental economics about the possible misspecification of monotone sQRE in situations in which there was no evidence of any issues with this theory. As a byproduct, we answer \citet{Goeree-Holt-Palfrey-2005-EE}'s standing question about the relationship of the empirical content of monotone sQRE and rQRE.  Since a Nash equilibrium that can be approximated by payoff monotone behavior can always be approximated by rQRE  \citep{Velez-Brown-2018-EE}, Theorem~\ref{Thm:Paradox2} implies that the game-wise empirical content of monotone sQRE and rQRE may differ when at least an agent has at least three actions available.

The remainder of the paper proceeds as follows. Sec.~\ref{Sec:Model} introduces definitions. Sec.~\ref{Sec-Results} presents our results. Sec.~\ref{Sec-Discussion} concludes and details the relationship of our work with that of \citet{Goeree-Holt-Palfrey-2005-EE} and \citet{Haile-et-al-2008}; discusses the implications of our results for the related randomly disturbed payoffs models of \citet{Harsanyi-1973-IJGT}; and finally discusses the relevance of our results for the refinement of Nash equilibrium.

\section{Definitions}\label{Sec:Model}

\subsection{Falsifiability and specification of a theory}

Our primitive is a set of random variables whose realizations are observable and of interest to a researcher.  Let~$\Ob$ be the set of \emph{all} joint distributions of observable random variables. A theory is an abstraction that the researcher creates describing the relationship between the joint distribution of a set of unobservable random variables that the researcher introduces in the analysis and the joint distribution of observable random variables.  Formally, a \textit{theory} is a pair $(\U,\mathcal{S})$ where~$\U$ is \emph{a} set of joint distributions of unobservable random variables and $\mathcal{S}$ is a correspondence $\mu\in\U\mapsto \mathcal{S}(\mu)\subseteq \Ob$. A \textit{model} of theory $(\U,\mathcal{S})$ is a pair $(\mu,\mathcal{S}(\mu))$ where $\mu\in\U$. The interpretation is that given the joint distribution on unobservable random variables is $\mu$, the theory determines that the joint distribution of observable random variables is necessarily in $\mathcal{S}(\mu)$. Whenever convenient we describe a theory by the collection of its models. A \textit{phenomenon} is a set of joint distributions on observables $\Phen\subseteq\Ob$. A theory  is \textit{well-specified} for the study of $\Phen$ if for each $\rho\in\Phen$ there is $\mu\in\U$ such that $\rho\in \mathcal{S}(\mu)$. A theory is \textit{consistent} with $\Phen$ if for each $\mu\in\U$, $\mathcal{S}(\U)\subseteq \Phen$. A theory is \textit{falsifiable} if $\mathcal{S}(\U)\subsetneq\Ob$.\footnote{Our definitions essentially follow \citet{Chambers-et-al-2014-AER}. We differ in that we explicitly define a theory as relating random distributions of unobservables and observables in the tradition in economics that assumes that even though the researchers are able to observe only finite data, when this data is powerful enough they can determine with some level of confidence whether distributions on observables belong to a particular  family of distributions \citep[e.g.][]{ATHEY-Haile-2007}. Thus, our notion of falsification should be understood as the existence of finite sets from which a researcher can conclude with a reasonable level of confidence that the distribution that generated the data does not belong to the distributions generated by the theory.}

\subsection{Observable Nash equilibrium}

Consider a set of agents $N\equiv\{1,...,n\}$ and a finite set of available actions to each agent,~$A_i$. Let $A\equiv A_1\times\dots\times A_n$.  Since we will not make any statement about environments with different set of agents and actions, let us fix $N$ and $A$. To avoid trivialities we assume that there are at least two agents with at least two actions available. For each $u\in \R^{N\times A}$ let  $\Gamma(u)\equiv(N,A,u)$ be the corresponding normal form game.

A strategy for agent $i$ is a probability distribution on $A_i$, denoted by $\sigma_i\in\Delta(A_i)$. A pure strategy places probability one on a given action. A strategy is interior if it places positive probability on each possible action. A profile of strategies is denoted by $\sigma\equiv(\sigma_i)_{i\in N}\in\Delta\equiv\Delta(A_1)\times\dots\times\Delta(A_n)$. Given $S\subseteq N$, we denote a subprofile of strategies for these agents by $\sigma_S$. When $S=N\setminus\{i\}$, we simply write $\sigma_{-i}\in \Delta_{-i}\equiv\times_{j\in N\setminus\{i\}}\Delta(A_j)$. Consistently, we concatenate partial strategy profiles as in $(\sigma_{-i},\mu_i)$. We consistently use this convention when operating with vectors.

We assume that both $u$ and~$\sigma$ are observable.  This is reasonable in laboratory experiments. This is also a valuable thought experiment that has been used to provide a foundation to Nash equilibrium theory \citep{Harsanyi-1973-IJGT}.

We denote agent~$i$'s expected utility given strategy profile $\sigma$ by $U_{u_i}(\sigma)$. We write $U_{u_i}(\sigma_{-i},a_i)$ for the utility that agent~$i$ gets from playing action~$a_i$ when the other agents play~$\sigma_{-i}$. A \textit{Nash equilibrium of $\Gamma(u)$} is a profile of strategies $\sigma$ such that for each $i\in N$ and each $\sigma'_i\in\Delta(A_i)$, $U_{u_i}(\sigma)\geq U_i(\sigma_{-i},\sigma'_i)$ \citep{Nash-1951}. We denote this set by $\Nash(\Gamma(u))$.

\subsection{Structural QRE theory}

Nash equilibrium theory, i.e., the one that associates with a trivial distribution of unobservables the set of all degenerate outcomes $\{(u,\sigma):u\in\R^{N\times A},\sigma\in \Nash(\Gamma(u))\}$, is easily \textit{falsified}. Because of this, researchers have modified this theory in order to account for a meaningful effect of unobservables. There are multiple ways in which this can be done. One can assume that the researcher imperfectly observes payoffs, that the agent is not a perfect utility maximizer, that the agent is not an expected utility maximizer, and so on. The structural QRE theory of \citet{mckelvey:95geb}, which we introduce next, allows the researcher to articulate the first two of these ideas while retaining the expected utility hypothesis.

For each $i\in N$ let $\Bcal_i$ be the set of Borel probability measures on~$\R^{A_i}$ that are absolutely continuous with respect to the Lebesgue measure and $\Bcal\equiv\Bcal_1\times\dots\times\Bcal_n$.\footnote{We do not assume perturbations have full support as \citet{mckelvey:95geb}. Because of this our sQRE models do not necessarily satisfy interiority and strictly contain \citet{mckelvey:95geb}'s sQRE models. Besides allowing us to present slightly more general results, this allows us to easily extend our analysis to \citet{Harsanyi-1973-IJGT}'s randomly disturbed payoff models as defined in a general form by \citet{GOVINDAN-Reny-Robson-2003-GEB}.} For each $\mu\equiv(\mu_i)_{i\in N}\in \Bcal$, let $(\Gamma(u),\mu)$ be the incomplete information game with independent common prior $\mu\equiv \mu_1\times\cdots\times\mu_n$ where payoffs are determined as follows. Given type $x_i\in \R^{A_i}$ for agent $i$, her expected utility index is $a\in A\mapsto u_i(a)+x_i(a_i)$. The interpretation of these perturbations is that the agent fails to perfectly recognize the difference of expected payoffs between the actions and correctly maximize \citep{mckelvey:95geb} or that there are unobserved shocks to expected utility of actions.

Absolute continuity of perturbations allows the researcher to associate, based on the expected utility maximization hypothesis, a unique observable behavior to agent~$i$, almost every $\mu_i$, as a response to a distribution of play of the other agents. More precisely, consider $i\in N$ and suppose that $\mu_i\in\Bcal_i$. One can easily see that for each $\sigma_{-i}\in\Delta_{-i}$ and for $\mu_i$ almost every realization of the perturbation, say $x_i$, there is a unique maximizer of
\begin{equation}a_i\in A_i\mapsto\sum_{a_{-i}}(u_i(a_{-i},a_i)+x_i(a_i))\sigma_{-i}(a_{-i})=U_{u_i}(\sigma_{-i},a_i)+x_i(a_i).\label{Eq:QRE-BR}\end{equation}
Thus, given $\mu_i$ and $\sigma_{-i}$, the probability with which agent $i$ is observed playing a given action is uniquely defined under the hypothesis of expected utility maximization.  Let
\[B^{u_i,\mu_i}_i(\sigma_{-i})\equiv (B^{u_i,\mu_i}_{ia_i}(\sigma_{-i}))_{a_i\in A_i}\in \Delta(A_i),\]
be this distribution. It is easy to see that this function is continuous and that the fixed points of $\sigma\in\Delta\mapsto(B^{u_i,\mu_i}_i(\sigma_{-i}))_{i\in N}\in \Delta$, which exist by Brouwer's fixed point theorem, are the set of Bayesian Nash equilibria of $(\Gamma(u),\mu)$, which we denote by $\BNE(\Gamma(u),\mu)$ \citep{mckelvey:95geb}.

Note from (\ref{Eq:QRE-BR}) that the best response operator depends only on the vector of expected utilities $(U_{u_i}(\sigma_{-i},a_i))_{a_i\in A_i}$. Thus, a best response operator is characterized by the function $x\in\R^{A_i}\mapsto Q_i^{\mu_i}(x)$, where
\[\sigma_{-i}\in \Delta_{-i}\mapsto B_i^{u_i,\mu_i}(\sigma_{-i})=Q_i^{\mu_i}(U_{u_i}(\sigma_{-i},\cdot))\in \Delta(A_{-i}).\]
The function $Q^\mu\equiv(Q^{\mu_i}_i)_{i\in N}$ only depends on $N$, $A$, and $\mu$, and does not depend on $u$. It is usually referred to as a \textit{structural Quantal Response Function} (sQRF) \citep{mckelvey:95geb,Goeree-Holt-Palfrey-2005-EE}.  The sQRF most commonly used in empirical analysis of experimental data is the Logistic form, $l^{\lambda_i}$, which is associated with the so-called double-exponential i.i.d.\ perturbation \citep{Goeree-et-al-2018}, and, for $\lambda_i\in[0,+\infty)$, assigns to each $a_i\in A_i$ and each $x\in\R^{A_i}$ the value,
\begin{equation}l^{\lambda_i}_{ia_i}(x)\equiv\frac{e^{\lambda_i x_{a_i}}}{\sum_{\hat a_i\in A_i}e^{\lambda_i x_{\hat a_i}}}.\label{Equation-Logistic-QRE}\end{equation}

\begin{definition}\rm The \textit{structural QRE (sQRE) theory} is the pair $(\Bcal,\mathcal{N})$ where for each $\mu\in\Bcal$, $\mathcal{N}(\mu)$ is the set of all $(u, \sigma)$ with $u\in \R^{N\times A}$ and $\sigma\in \BNE(\Gamma(u),\mu)$.
\end{definition}

\section{Results}\label{Sec-Results}

In an effort to construct an econometric model based on sQRE, \citet{Haile-et-al-2008} study a particular $\Psi_{HHK}\subseteq \Bcal$ that includes all perturbations that exhibit certain form of correlation. These authors show that $(\Psi_{HHK},\Ncal)$ is not \textit{falsifiable}, i.e., for each $(u,\sigma)\in\R^{N\times A}\times \Delta$ there is $\mu\in\Psi_{HHK}$ such that $\sigma\in \BNE(\Gamma(u),\mu)$. Thus, $(\Bcal,\Ncal)$ is not \textit{falsifiable}.

\textit{Falsifiability} is arguably a defining characteristic of scientific theories \citep{Popper-1965}. Thus, \citet{Haile-et-al-2008}'s result showed the need to discipline sQRE. As a response to this, \citet{Goeree-Holt-Palfrey-2005-EE} proposed to require \textit{consistency} of sQRE with an observable property of behavior that was always implicitly assumed in empirical applications of these models.

\begin{definition}\rm Let $u\in\R^{N\times A}$ and $\sigma\equiv(\sigma_i)\in \Delta$. $(u,\sigma)$ is \textit{payoff monotone} if for each $i\in N$, $\sigma_i$ is ordinally equivalent with $(U_{u_i}(\sigma_{-i},a_i))_{a_i\in A_i}$. We denote the set of $(u,\sigma)$ that are \textit{payoff monotone} by $\M$. When $(u,\sigma)\in\M$, we say that $\sigma$ is \textit{payoff monotone} for $u$.
\end{definition}

Intuitively, \textit{payoff monotonicity} requires that frequencies of play reveal the ranking of expected payoffs of all actions.

\begin{table}[t]
  \centering
  \begin{tabular}{ccccc}
     & &\multicolumn{3}{c}{Player 1}
  \\
 &&  \multicolumn{1}{c}{$a_1$}&\multicolumn{1}{c}{$a_2$}&\multicolumn{1}{c}{$a_{3}$}
\\\cline{3-5}
 Player 2&$b_{1}$&\multicolumn{1}{|c}{$100,100$}&\multicolumn{1}{|c|}{$100,100$}&\multicolumn{1}{c|}{$100,100$}
\\\cline{3-5}
&$b_2$&\multicolumn{1}{|c}{$0,101$}&\multicolumn{1}{|c|}{$0,102$}&\multicolumn{1}{c|}{$0,104$}\\\cline{3-5}
  \end{tabular}
    \caption{Game $\Gamma(u^*)$.}\label{Tab:Gamma1}
\end{table}

\begin{example}\rm\label{Ex:prism}
Consider $u^*$ defined in Table~\ref{Tab:Gamma1}.  The strategy space in $\Gamma(u^*)$ is a prism (Fig.~\ref{Fig:prism}) and $\Nash(\Gamma(u^*))$ is the set of profiles in which player~$2$ chooses $b_1$ and player~$1$ arbitrarily randomizes among her three available actions. In Fig.~\ref{Fig:prism}, this set corresponds to the upper lid of the prism.

Consider a profile $\sigma\in \Delta(A)$. If $\sigma_2(b_1)=1$, the expected utility of player $1$ is the same for each of her actions. Thus, only the centroid of the upper lid of the prism is payoff monotone for $u^*$. Since $b_1$ strictly dominates $b_2$ for player~$2$, then in a payoff monotone distribution for the game, $\sigma_2(b_1)>1/2$. Now, suppose that $\sigma_2(b_2)>0$. Then, $U_{u_1^*}(\sigma_2,a_3)>U_{u_1^*}(\sigma_2,a_2)>U_{u_1^*}(\sigma_2,a_1)$. Thus, in a payoff monotone distribution for $u^*$ in which $\sigma_2(b_2)>0$, we must have $\sigma_1(a_3)>\sigma_1(a_2)>\sigma_1(a_1)$. Thus, the set of payoff monotone distributions for $u^*$ contains, in particular, the centroid of the upper lid of the prism of strategies and the interior of the wedge area formed by the two shaded subprisms in Fig.~\ref{Fig:prism}. The subset of $\Nash(\Gamma(u^*))$ that can be approximated by payoff monotone behavior for $u^*$ is the union of the upper lids of the shaded subprisms in Fig.~\ref{Fig:prism}. $\qed$

\medskip

Game $\Gamma(u^*)$ is deceivingly simple. Trembling based refinements, for instance, predict that the only plausible distribution of play is that Player~$1$ plays $a_3$ and player~$2$ plays $b_1$. The reasoning is that since $a_3$ weakly dominates the other actions for player $1$, this agent should preemptively play $a_3$ in case player $2$ makes a mistake. Experiments with dominant strategy games, which generate payoff matrices similar to $u^*$, e.g., the second price auction, point to the prevalence of weakly dominated behavior in these games, however \citep[see][for a meta-study]{Velez-2018-SP}. Essentially, experiments have shown is that if player~$2$ (types) notices the difference in payoffs between her actions and plays $b_1$ with high probability, then player~$1$ is almost indifferent between her three actions. Since players do not seem to react to differences in expected payoffs by not playing at all suboptimal actions, it is plausible to observe frequencies of play in which player~$2$ is careful and player~$1$ is careless. (For instance, in the second-price auction experiments of \citet{Kegel-Levin-1993-EJ} and \citet{Li-AER-17}, which hold some similarity with game $\Gamma(u^*)$, dominant strategy play is consistently below 30\%.) Behavior is not arbitrary, however. Players tend to inform their choices based on the expected payoffs \citep{Goeree-Holt-Palfrey-2016-Book}. Thus, one can expect that player $1$ plays $a_3$ with higher probability than $a_2$ and $a_2$ with higher probability than~$a_1$. If one cannot expect further restrictions from data than those imposed by this intuition, each payoff monotone distribution for $u^*$, say $\sigma_2(b_1)=0.99$, $\sigma_1(a_1)=0.03$, $\sigma_1(a_2)=0.35$, and $\sigma_1(a_3)=0.62$, should be in the range of a theory that is \textit{well-specified} to study this phenomenon.
%\newpage
\end{example}
\begin{center}
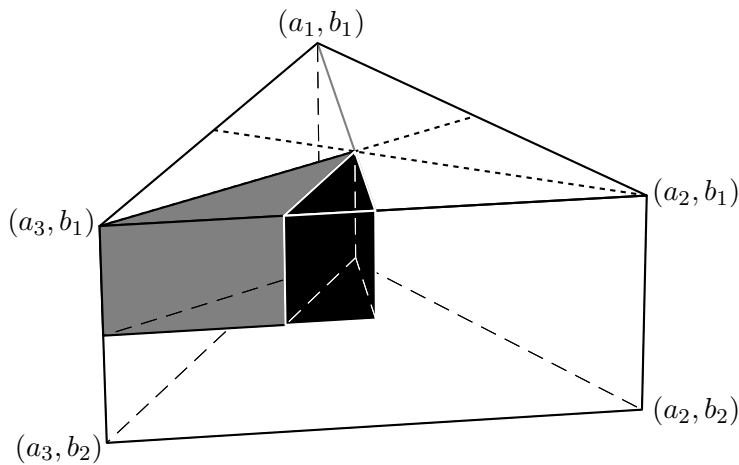
\begin{figure}[t]
\centering
\begin{pspicture}(-4,-1.5)(3,5)%\grilla
 \psset{viewpoint=50 20 25,Decran=60}
\psSolid[
object=line,args= -3.4641016151 -3 3 1.7320508076 0 3,linecolor=gray
]
\psSolid[
object=line,args=3.4641016151 -3 3 -1.7320508076 0 3,linestyle=dotted
]
\psSolid[
object=line,args=0 3 3  0 -3 3,linestyle=dotted
]
\psSolid[
object=prisme,
action=draw,
%linecolor=red,
base=-3.4641016151 -3 3.4641016151 -3 0 3,
h=3]
\psSolid[
object=prisme,
action=draw*,
%linecolor=white,
fillcolor=gray,
base=0 -1 3.4641016151 -3 2.3094010768 -1,
h=1.5](0,0,1.5)
\psSolid[
object=prisme,
action=draw*,
linecolor=white,
fillcolor=black,
base=0 -1 2.3094010768 -1 1.7320508076 0,
h=1.5](0,0,1.5)
\rput[l](3,2.4){$(a_2,b_1)$}
\rput[l](3,-.5){$(a_2,b_2)$}
\rput[r](-4.3,-1){$(a_3,b_2)$}
\rput[r](-4.4,2){$(a_3,b_1)$}
\rput[c](-1.4,4.65){$(a_1,b_1)$}

\end{pspicture}
\caption{Strategies space in game $\Gamma(u^*)$, i.e., $\Delta(\{a_1,a_2,a_3\})\times\Delta(\{b_1,b_2\})$. There is a one to one correspondence between the strategies space and the prism.  For a given point in the prism, the probability with which player~$1$ plays an action is the distance of the point to the face opposed to the vertices representing the profiles in which the agent plays the action. The probability with which player~$2$ plays~$b_1$ is the distance of the point to the lower base of the prism. For instance, the centroid of the upper lid of the prism is the profile in which player~$1$ uniformly randomizes among all actions and player~$2$ plays $b_1$ with probability one. The pure strategy profile $(a_1,b_2)$ corresponds to the hidden vertex on the back of the prism. The upper lid of the prism is the set of Nash equilibria of the game. The interior of the shaded subprisms belong to the set of payoff monotone distributions $u^*$ (this set also contains the centroid of the upper lid of the strategies space and the distributions such that $1/2>\sigma_2(b_2)>0$ and $\sigma_1(a_1)=0<\sigma_1(a_2)<\sigma_1(a_3)$). The upper lids of the shaded subprisms are the Nash equilibria that are in the closure of weakly payoff monotone distributions for $u^*$. The empirical content of any monotone structural QRE model for this game is contained in the gray subprism. Thus, with exception of its left face, no element of the black subprism is in the range of any monotone sQRE.}
\label{Fig:prism}
\end{figure}
\end{center}
\vspace{-1.1cm}

%\newpage

%\medskip
\textbf{Example~\ref{Ex:prism} (continuation)}: The logistic QRE model, including its heterogeneous parameter versions, does not span the whole space of payoff monotone distributions for~$u^*$.  Indeed, no distribution such that $\sigma_1(a_2)>1/3$ is in the range of this model. In particular, no profile of strategies in the interior of the black subprism in Fig.~\ref{Fig:prism} is in the range of logistic QRE models. To prove this, let $\alpha\equiv \sigma_1(a_2)$ and $x\equiv\sigma_1(a_1)$. Given the functional form of the logistic QRF, we know that $\alpha/x=\exp(\lambda_i\sigma_1(b_2))$ and $\sigma_1(a_3)/\alpha=(1-\alpha-x)/\alpha=\exp(2\lambda_i\sigma_1(b_2))$. Thus, $\sigma_1(a_3)>\alpha^2/x$ and $\sigma_1(a_3)+\sigma_1(a_2)+\sigma_1(a_2)>\alpha^2/x+\alpha+x$. It is easy to see that the minimum of $(0,1-2\alpha]\mapsto\alpha^2/x+\alpha+x$ is achieved at $x=1-2\alpha$ and that this minimum is always greater than one for $\alpha>1/3$. $\qed$

\medskip
It is not surprising that the logistic QRE has these limitations. The logistic QRE model is tied to very particular perturbations of payoffs. For instance, the logistic sQRF has the property that the ratio between the probability of two actions is a function of the difference between their expected payoffs. In our analysis of $\Gamma(u^*)$, this property is seemingly responsible for the limitations of the model.

Thus, the question that remains is whether the limitation of logistic QRE models is simply an artifice of the form of the specific perturbations behind them or is a characteristic of all monotone sQRE. Surprisingly, the answer is the second. These limitations are shared by \emph{all} monotone sQRE (Theorem~\ref{Thm:Paradox2} below is proved by means of the analysis of a general family of games containing Example~\ref{Ex:prism}). This has the far reaching consequence that sQRE theory cannot be disciplined with $\M$ and remain \textit{well-specified} for the study of this phenomenon. Moreover, this is not tied to the number of players in our example or the number of actions they have available, as long as at least an agent has at least three actions available.\footnote{The empirical content of rQRE and sQRE for a given game is the same when all agents have at most two actions available \citep{Goeree-Holt-Palfrey-2005-EE}.}

\begin{theorem}[The paradox of monotone sQRE]\rm\label{Thm:Paradox}Suppose that at least an agent has at least three actions available. Then,
no set of sQRE models can be both \textit{consistent} with $\M$ and \textit{well-specified} to study this phenomenon.
\end{theorem}

The limitations of the monotone sQRE models are actually starker than what Theorem~\ref{Thm:Paradox} reveals. In some environments, one can also expect that behavior may eventually approach mutual best responses. Thus, it may be reasonable to be interested in models that generate behavior in a certain proximity of the set of Nash equilibria. As Example~\ref{Ex:prism} also hints, the incompatibility persists even if one only requires that the theory be able to explain payoff monotone behavior that is close to observable Nash equilibria: The upper lid of the black subprism in Fig.~\ref{Fig:prism} is part of $\Nash(\Gamma(u^*))$; each of these profiles belongs to an open set (relative to the strategies space) that is outside of the range of each monotone sQRE model, however.

\begin{definition}\rm Let $\Mcal\subseteq \Bcal$ be the set of perturbations $\mu$ for which for each $u\in\R^{N\times A}$ and each $\sigma\in \BNE(\Gamma(u),\mu)$, $(u,\sigma)\in\M$.
\end{definition}

\begin{theorem}\label{Thm:Paradox2}\rm Suppose that at least an agent has at least three actions available. Then, there is $u\in \R^{N\times A}$ for which there are $\sigma^*\in \Nash(\Gamma(u))$ and $\varepsilon>0$ such that
\begin{enumerate}
\item $\sigma^*$ belongs to the closure of $\{\sigma:(u,\sigma)\in\M\}$, and
\item $\{\sigma:||\sigma-\sigma^*||<\varepsilon\}\cap \{\sigma:\exists\mu\in\Mcal,\textrm{ s.t. }\sigma\in \BNE(\Gamma(u),\mu)\}=\emptyset$.
\end{enumerate}
\end{theorem}

We discuss the proof of Theorem~\ref{Thm:Paradox2} (this result implies Theorem~\ref{Thm:Paradox}). Our first step is to realize that \textit{consistency} of a sQRE model with \textit{payoff monotonicity} is equivalent to monotonicity of the corresponding sQRF.

\begin{lemma}\label{Lem:msQRE=mQRF}\rm Let $\mu\in\Bcal$. The following are equivalent.
\begin{enumerate}
\item For each $u\in\R^{N\times A}$ and each $\sigma\in\BNE(\Gamma(u),\mu)$, $(u,\sigma)\in\M$.
\item For each $i\in N$ and each  $x_i\in\R^{A_i}$, $Q^{\mu_i}_i(x_i)$ is ordinally equivalent to $x_i$.
\end{enumerate}
\end{lemma}

Proving Theorem~\ref{Thm:Paradox2} involves constructing a utility profile that admits a Nash equilibrium that is close to payoff monotone behavior and for which no monotone sQRE induces behavior close to that equilibrium. In order to get intuition on how to achieve this, it is useful to prove a much less ambitious statement. It is well known that perturbations satisfying the following property induce behavior in $\M$ \citep{mckelvey:95geb,Goeree-Holt-Palfrey-2016-Book}. This property is satisfied by i.i.d.\ perturbations, as in the Logistic sQRE model.

\begin{definition}\rm $\mu\in\Bcal$ is \textit{permutation invariant} if for each $i\in N$, each permutation $\pi:\R^{A_i}\rightarrow\R^{A_i}$, and each measurable $C\subseteq \R^{A_i}$, $\mu(C)=\mu(\pi(C))$.
\end{definition}

We now show that in Example~\ref{Ex:prism}, no distribution in the interior of the black subprism in Fig.~\ref{Fig:prism} is in the range of any permutation invariant sQRE for this game.

\medskip
\textbf{Example~\ref{Ex:prism} (Continuation)}: Suppose that $\sigma_{2}(b_2)>0$. Thus, $U_{u_1}(\sigma_{2},a_2)-U_{u_1}(\sigma_{2},a_1)<U_{u_1}(\sigma_{2},a_3)-U_{u_1}(\sigma_{2},a_2)$.  Let $\sigma_1\in\Delta(A_1)$ be such that $0<\sigma_1(a_1)<1/3<\sigma_1(a_2)<\sigma_1(a_3)$. Clearly, $\sigma_1$ is ordinally equivalent to $U_{u_1}(\sigma_{2},\cdot)$.  We claim that there is no permutation invariant $\mu$ such that $\sigma_1=Q^{\mu_1}_1(U_{u_1}(\sigma_{2},\cdot))$. Even though it would not belong to $\Bcal$, it is enough to think of $\mu_1$ as the realization of a fair dice with six faces labeled by the permutations of a vector $(0,x,y)$ where $0<x<y$ and estimate the probability that action $a_2$ is chosen if the perturbation is determined by the dice.\footnote{Our appendix presents a formal proof of Theorem~\ref{Thm:Paradox2} that is not based on permutation invariance. Thus, the informality of our argument here allows us to gain intuition without entering into technicalities.} More precisely, when the dice falls on $(0,x,y)$ the utility of action $a_1$ gets a perturbation $0$, the utility of action $a_2$ gets a perturbation $x$, and the utility of action $a_3$ gets a perturbation $y$, and so on for the other faces of the dice. Obviously, when the dice falls on $(0,x,y)$ or $(x,0,y)$ agent~$1$'s maximizer is action $a_3$. If the dice falls on $(y,0,x)$, the maximizer is never $a_2$. If action $a_2$ is to be the maximizer with positive probability, $a_2$ must be the maximizer when the dice falls on $(x,y,0)$.  There are only two additional faces of the dice for which $a_2$ can be the maximizer, $(0,y,x)$ or $(y,x,0)$. Suppose that $a_2$ is the maximizer for $(y,x,0)$. Then, $x+U_{u_1}(\sigma_{2},a_2)>y+U_{u_1}(\sigma_{2},a_1)$. Thus, $x+U_{u_1}(\sigma_{2},a_3)>y+U_{u_1}(\sigma_{2},a_2)$. This means that $a_2$ is not the maximizer for $(0,y,x)$. Thus, $a_2$ can actually be the maximizer in at most two out of six possible outcomes of the dice roll. $\qed$

\medskip
This analysis reveals that permutation invariant sQRE models are closely tied to the cardinal information in observable utilities. Indeed, these models impose restrictions on the difference of the probability assigned to two actions based on the difference in expected utility between these actions.

Our proof of Theorem~\ref{Thm:Paradox2} is by means of  an example that extends the structure of $u^*$ to general action spaces. The subtlety of our analysis, compared to the simplicity of the logistic QRE and the permutation invariant sQRE cases, is derived from our need to establish relationships between the outcomes of a sQRF based only on \textit{payoff monotonicity}.

\begin{table}[t]
  \centering
  \begin{tabular}{cp{2.5cm}p{2.5cm}p{2.5cm}p{2.5cm}p{2.5cm}}
      &\multicolumn{5}{c}{Player 1}
  \\
  &  \multicolumn{1}{c}{$a_1$}&\multicolumn{1}{c}{$a_2$}&\multicolumn{1}{c}{$\dots$}&\multicolumn{1}{c}{$a_{K-1}$}&
  \multicolumn{1}{c}{$a_{K}$}
\\\cline{2-6}
$a_{-1}^*$&\multicolumn{1}{|c}{$5$}&\multicolumn{1}{|c|}{$5$}&\multicolumn{1}{|c|}{$\dots$}
&\multicolumn{1}{|c|}{$ 5$}&\multicolumn{1}{|c|}{$5$}
\\\cline{2-6}
$A_{-1}\setminus\{a_{-1}^*\}$&\multicolumn{1}{|c}{$1$}&\multicolumn{1}{|c|}{$2$}&\multicolumn{1}{|c|}{$\dots$}&\multicolumn{1}{|c|}{$ 2$}&\multicolumn{1}{|c|}{$4$}\\\cline{2-6}
  \end{tabular}
    \caption{Game $\Gamma(u)\equiv(N,A,u)$, $N\equiv\{1,...,n\}$, $A_1\equiv\{a_1,...,a_K\}$ with $K\geq 3$, and $|A_{-1}|\geq 2$. The table shows the payoff of agent~$1$. Each agent $j>1$ gets a payoff of $1$ if she plays $a^*_j$ and zero otherwise.}\label{Tab:Upsilon1}
\end{table}

The game that allows us to prove Theorem~\ref{Thm:Paradox2} is defined in Table~\ref{Tab:Upsilon1}. In this game each agent $i\neq 1$ has a strictly dominant action. The profile of these strictly dominant actions for these agents is~$a_{-1}^*$. Agent~$1$ is indifferent among all actions if all other agents play their strictly dominant action. Agent~$1$ has three different types of actions. Action~$a_1$, which is weakly dominated by actions~$a_2,...,a_{K-1}$, which are all payoff equivalent. All actions $\{a_1,...,a_{K-1}\}$ are weakly dominated by~$a_K$ for this agent. The essential feature of this game is that given any $\sigma_{-1}$ for which $\sigma_{-i}(a^*_{-i})<1$, the difference in agent~$1$'s expected payoff between actions~$a_K$ and~$a_{K-1}$ is greater than the difference in expected payoff between actions~$a_2$ and~$a_1$.

\begin{lemma}\label{Lm:EUpsilon1}\rm Let $\Gamma(u)$ be the game in Table~\ref{Tab:Upsilon1}. There is  $\sigma\in \Nash(\Gamma(u))$ that belongs to the closure of $\{\gamma:(u,\gamma)\in\M\}$ and in which each agent $j\neq 1$ plays the strictly dominant action and $\sigma_1(a_1)<1/K<\sigma_1(a_2)= \dots=\sigma_1(a_{K-1})<\sigma_1(a_{K})$.
\end{lemma}

Lemma~\ref{Lm:EUpsilon1} states that there are $(u,\sigma)\in\M$ arbitrarily close to an observable equilibrium of $\Gamma(u)$ in which agent~$1$ plays actions $\{a_2,...,a_{K-1}\}$ with probability greater than $1/K$. The following proposition identifies restrictions on distributions generated by monotone sQRE models. The proof of Theorem~\ref{Thm:Paradox2} is completed by showing, based on these restrictions, that it is impossible for these models to generate behavior close to the equilibrium identified in Lemma~\ref{Lm:EUpsilon1}.

\begin{proposition}\label{Pro:k3-or-permut--nonmonot1}\rm Let $\Gamma(u)$ be the game in Table~\ref{Tab:Upsilon1} and $\mu\in\Mcal$. Let $\sigma_{-i}\in\Delta(A_{-i})$ be such that $U_{u_i}(\sigma_{-i},a_1)<U_{u_i}(\sigma_{-i},a_2)=\dots=U_{u_i}(\sigma_{-i},a_{K-1})< U_{u_i}(\sigma_{-i},a_K)$,
and $U_{u_i}(\sigma_{-i},a_K)-U_{u_i}(\sigma_{-i},a_{K-1})>U_{u_i}(\sigma_{-i},a_{2})-U_{u_i}(\sigma_{-i},a_{1})$. Then, $B^{u_i,\mu_i}_{ia_{K-1}}(\sigma_{-i})\leq1/K$.
\end{proposition}

\section{Discussion}\label{Sec-Discussion}

We have shown that the sQRE theory of \citet{mckelvey:95geb} cannot be disciplined by \textit{consistent} with payoff monotonicity and remain \textit{well-specified} to study this phenomenon.

As outlined in the introduction, this result points to the need to understand better the empirical content of payoff monotone sQRE models. It is interesting, but beyond the scope of this paper, to develop a characterization of the empirical content of payoff monotone sQRE models. Such a characterization is the basis to determine if there is empirical support for the restrictions, beyond payoff monotonicity, that are implicit in payoff monotone sQRE. Alternatively, one can also reevaluate the structure of the unobservables in sQRE models in order to construct a thory that is both \textit{consistent} with $\M$ and \textit{well-specified} to study this phenomenon. In a companion paper we show that the rQRE models of \citet{Goeree-Holt-Palfrey-2005-EE}, and in particular a finitely parametric form of these models based on the control costs games  of \citet{VanDamme-1991-Springer}, essentially have these properties \citep{Velez-Brown-2018-EE}.\footnote{Regular QRE satisfies interiority, thus it may not generate some payoff monotone distributions in the boundary of the strategies space.}

sQRE models are closely related with \citet{Harsanyi-1973-IJGT}'s randomly disturbed payoff models, which only differ in that unobservables are perturbations $x_i\in\R^A$ that determine utility indices $a\in A\mapsto u_i(a)+x_i(a)$.\footnote{Note that in sQRE models perturbations for agent $i$ are functions of agent~$i$'s action only.} Proposition~\ref{Pro:k3-or-permut--nonmonot1} generalizes to this class of models, with the proviso that one only requires best response functions be ordinally equivalent to payoff vectors  (we provide details in the Appendix). Lemma~\ref{Lem:msQRE=mQRF} does not directly generalize. Observe that this result is a statement about QRFs, which are not well-defined in \citet{Harsanyi-1973-IJGT}'s models. We do not know if an equivalent result actually holds. Without this lemma, the implications of Lemma~\ref{Lm:EUpsilon1} and Proposition~\ref{Pro:k3-or-permut--nonmonot1}, which both hold for these models, can be summarized as follows. \textit{If one requires best response operators to produce only behavior that is ordinally equivalent to observed expected utility, \citet{Harsanyi-1973-IJGT}'s models are not well-specified for the study of $\M$}. Note that the first requirement in this statement is a decision theoretical thought experiment, as opposed to the requirement of consistency with an observable phenomenon in strategic situations as in Theorem~\ref{Thm:Paradox}. Thus, it is still an open question whether the equivalent to Theorem~\ref{Thm:Paradox} holds for \citet{Harsanyi-1973-IJGT}'s models.

One can use randomly perturbed payoff models to refine Nash equilibrium \citep{VanDamme-1991-Springer}. Essentially, one can identify a Nash equilibrium as implausible if it cannot be approached by behavior in perturbed games. If one does not restrict perturbations, one may be using for this exercise models that are easily rejected by data. Thus, it is sensible to impose some discipline on perturbations. \citet{VanDamme-1991-Springer} shows that if one requires permutation invariance of perturbations on approximations by \citet{Harsanyi-1973-IJGT}'s models, the set of Nash equilibria that can be approximated, the ``firm equilibria,'' may be a strict subset of the Nash equilibrium set. \citet{Mackelvey-Palfrey-1996-JER} alternatively propose to refine Nash equilibria by approximation of Logistic QRE. Both of these approximations implicitly assume payoff monotonicity. As we have shown, randomly perturbed payoff models depend on assumptions about the structure of unobservables that are beyond consistency with payoff monotonicity. Thus, one can think of a better founded refinement of Nash equilibrium by only requiring approximation by payoff monotone behavior. This is the ``empirical equilibrium'' refinement of \citet{Velez-Brown-2018-EE}. Interestingly, Theorem~\ref{Thm:Paradox2} implies that there are empirical equilibria that are neither firm, nor approachable by any structural QRE model. Thus, both these last two refinements impose restrictions due to their particular functional forms, which are not observable, as opposed to empirical equilibrium which is based only on consistency with payoff monotonicity.

\section*{Appendix}

\begin{proof}[Proof of Lemma~\ref{Lem:msQRE=mQRF}] Let $\mu\in\Bcal$ and $u\in\R^{N\times A}$. Consider first $\sigma\in \BNE(\Gamma(u),\mu)$. Then, $\sigma_i=Q^{\mu_i}_i(U_{u_i}(\sigma_{-i},\cdot))$. If statement 2 is holds, then $(u,\sigma)\in\M$. Suppose now that for each $u\in\R^{N\times A}$ and each $\sigma\in\BNE(\Gamma(u),\mu)$ we have that $(u,\sigma)\in\M$. We claim that statement 2 holds. Suppose by contradiction that there are $i\in N$ and $x_i\in\R^{A_i}$ such that $\sigma_i\equiv Q^{\mu_i}_i(x_i)$ is not ordinally equivalent with $x_i$. Let $u_{-i}\in\R^{N\setminus\{i\}\times A}$. Consider the function
\[\gamma_{-i}\in\Delta_{-i}\mapsto (Q_j^{\mu_j}(U^u_j(\gamma_{N\setminus\{i,j\}},\sigma_i,\cdot)))_{j\in N\setminus\{i\}}\in\Delta_{-i}.\]
This function is continuous because is the composition of structural QRFs and the expected utility operator. Thus, by Brouwer's fixed point theorem it has a fixed point $\sigma_{-i}$. For each $a\in A$ let $u_i(a)\equiv x_i(a_i)$. Then, $U_{u_i}(\sigma_{-i},\cdot)=x_i$. Thus, $\sigma=Q^{\mu_j}_j(U_{u_j}(\sigma_{-j},\cdot))_{j\in N}$. Thus, $\sigma\in \BNE(\Gamma(u),\mu)$.
\end{proof}

\begin{proof}[\textit{Proof of Lemma~\ref{Lm:EUpsilon1}}]Clearly $\Nash(\Gamma(u))$ is the set of distributions in which each $j\neq1$ plays the dominant action and agent $1$ arbitrarily randomizes. Let $\sigma$ be such that each agent $j\neq 1$ plays the strictly dominant action with certainty, and $\sigma_1(a_1)<\sigma_1(a_2)= \dots=\sigma_1(a_{K-1})<\sigma_1(a_{K})$. Then, $\sigma\in \Nash(\Gamma(u))$. Let $\lambda\in\N$ and $\sigma^\lambda$ be the convex combination that places $(1-1/\lambda)$ weight on $\sigma$ and $1/\lambda$ on a uniform distribution. Clearly as $\lambda\rightarrow\infty$, $\sigma^\lambda\rightarrow\sigma$. Thus, there is $\Lambda\in\N$ such that for each $\lambda\geq \Lambda$, $\sigma^\lambda$ is ordinally equivalent to $\sigma$. Since for each $\lambda\in \N$, $\sigma^\lambda$ is interior, $U_{u_i}(\sigma^\lambda_{-i},a_1)<U_{u_i}(\sigma^\lambda_{-i},a_2)= \dots=U_{u_i}(\sigma^\lambda_{-i},a_{K-1})<U_{u_i}(\sigma^\lambda_{-i},a_K)$, and for each $j\neq i$, if $a^*_j\in A_j$ is this agent's dominant action, $U_{u_j}(\sigma^\lambda_{-j},a_j^*)>U_{u_j}(\sigma^\lambda_{-j},a_j)$, and for each pair of actions $\{a_j,a_j'\}\subseteq A_j$ that are not dominant, $U_{u_j}(\sigma^\lambda_{-j},a_j)=U_{u_j}(\sigma^\lambda_{-j},a_j')$. Thus, $(u,\sigma^\lambda)\in\M$. Thus, $\sigma$ belongs to the closure of $\{\gamma:(u,\gamma)\in\M\}$. Thus, for each $1/K<\alpha<1/(K-1)$, there is $\sigma\in \Nash(\Gamma(u))$ that belongs to the closure of $\{\gamma:(u,\gamma)\in\M\}$ and such that $0=\sigma_1(a_1)<\alpha=\sigma_1(a_2)= \dots=\sigma_1(a_{K-1})<(1-(K-2)\alpha)=\sigma_1(a_{K})$.
\end{proof}

\begin{proof}[\textit{Proof of Proposition~\ref{Pro:k3-or-permut--nonmonot1}}]Let $k\in\{1,...,K\}$.\footnote{Proposition~\ref{Pro:k3-or-permut--nonmonot1} generalizes to all of \citet{Harsanyi-1973-IJGT}'s randomly disturbed payoff models for which best response correspondences are monotone with respect to expected utlity. In order to prove this results one can proceed as follows (see \citet{GOVINDAN-Reny-Robson-2003-GEB} for a modern presentation of this model that embeds both sQRE and \citet{Harsanyi-1973-IJGT}'s models). For each $x\in\R^A$, write $x_{a_k}\equiv (x_{(a_{-i},a_l)})_{a_{-i}\in A_{-i}}\in \R^{A_{-i}}$. Let $y_k(\sigma_{-i})=\sigma_{-i}\cdot x_{a_k}$. Apply the argument to the sets $X_k\equiv\{x\in\R^A:\{k\}={\arg\max}_{k=1,...,K}y_k(\sigma_{-i})\}$ and $G\equiv\{x\in\R^{A}:\{K-1\}=\arg\max_{l=1,...,K}v_l+y_{l}(\sigma_{-i})\}$, and adjust the definitions of all sets of perturbations to be defined by conditions on the components of $y(\sigma_{-i})$ instead of $x$.} Let $v_k\equiv U_{u_i}(\sigma_{-i},a_k)$ and
\[X_k\equiv\{x\in\R^{A_i}:\{k\}={\arg\max}_{l=1,...,K}x_{l}\}.\]
Let $\mu\in\Mcal$. Consider $\bar u_i\in\R^{A_i}$ for which agent $i$ has equal payoff from each action profile. Since $\mu\in\Mcal$, $B^{\bar u_i,\mu_i}_i(\sigma_{-i})=(1/K,...,1/K)$. Thus, for each $k=1,...,K$, $\mu_i(X_k)=1/K$.

Since $\mu\in\Bcal$, for each measurable set $G\subseteq\R^{A_i}$,
\begin{equation}\mu_i(G)=\sum_{l=1}^K\mu_i(G\cap X_k).\label{Eq:basis1}\end{equation}
Let $G\equiv\{x\in\R^{A_i}:\{K-1\}=\arg\max_{l=1,...,K}v_l+x_{l}\}$. Then, $B^{u_i,\mu_i}_{ia_{K-1}}(\sigma_{-i})=\mu_i(G)$. Since $v_1<v_2=\dots=v_{K-1}< v_K$,
\begin{equation}\mu_i(G)=\mu_i(G\cap X_1)+\mu_i(G\cap X_{K-1}).\label{Eq:g-expression1}\end{equation}
Let $D_1\equiv v_{2}-v_{1}$ and $D_2\equiv v_K-v_{K-1}$.   Consider $u'_i\in\R^{A_i}$ for which $y'\equiv U_{u_i'}(\sigma_{-i},\cdot)$ is such that $y_1'=\dots=y_{K-1}'<y_{K}'\equiv y_{K-1}'+D_2$ (one can simply make payoffs be independent of the action of the other agents). Since $\mu_i\in\Mcal$, $B^{u_i',\mu_i}_{ia_{1}}(\sigma_{-i})=B^{u_i',\mu_i}_{ia_{K-1}}(\sigma_{-i})$. Moreover,
\[\begin{array}{l}B^{u_i',\mu_i}_{ia_{K-1}}(\sigma_{-i})=\mu_i\left(\{x\in X_{K-1}:x_{{K-1}}>x_{K}+D_2\}\right),\\
B^{u_i',\mu_i}_{ia_1}(\sigma_{-i})=\mu_i\left(\{x\in X_{1}:x_{{1}}> x_{K}+D_2\}\right).\end{array}\]
Thus,
\[\mu_i\left(\{x\in X_{K-1}:x_{{K-1}}>x_{K}+D_2\}\right)=\mu_i\left(\{x\in X_{1}:x_{{1}}> x_{K}+D_2\}\right).\]
Since $\mu_i\in\Bcal_i$ and $\mu_i(X_1)=\mu_i(X_{K-1})$,
\begin{equation}\mu_i\left(\{x\in X_{K-1}:x_{{K-1}}<x_{K}+D_2\}\right)=\mu_i\left(\{x\in X_{1}:x_{{1}}< x_{K}+D_2\}\right).\label{Eq:reverse1}\end{equation}
We claim that
\begin{equation}\begin{array}{l}\mu_i\left(\{x\in X_{1}:\max\{x_1-D_1,x_2,...,x_{K-2},x_K\}<x_{{K-1}}\}\right)\\=\mu_i\left(\{x\in X_{1}:\max\{x_1-D_1,x_2,...,x_{K-1}\}<x_{{K}}\}\right).\end{array}\label{Eq:equaldist1}\end{equation}

Consider  $u_i''\in\R^{A_i}$ for which $y''\equiv U_{u_i''}(\sigma_{-i},\cdot)$ is such that $y_1''<y_2''=\dots=y_{K}''\equiv y_{1}''+D_1$.  Observe that
\[B^{u_i'',\mu_i}_{ia_{K-1}}(\sigma_{-i})=\mu_i\left(\left\{x\in X_{1}:\max\{x_1-D_1,x_2,...,x_{K-2},x_K\}<x_{{K-1}}\right\}\right)+\mu_i(X_{K-1}),\]
and
\[B^{u_i'',\mu_i}_{ia_K}(\sigma_{-i})=\mu_i\left(\left\{x\in X_{1}:\max\{x_1-D_1,x_2,...,x_{K-1}\}<x_{{K}}\right\}\right)+\mu_i(X_{K}).\]
Since $\mu\in\Mcal$, $B^{u_i'',\mu_i}_{ia_K}(\sigma_{-i})=B^{u_i'',\mu_i}_{ia_{K-1}}(\sigma_{-i})$. Thus, since $\mu_i(X_{K-1})=\mu_i(X_{K})$, (\ref{Eq:equaldist1}) follows.

Since $D_2\geq D_1$, by monotonicity of measures with respect to set inclusion
\[\begin{array}{l}\mu_i\left(\{x\in X_{1}:x_{1}<x_K+D_2\}\right)\\
\geq \mu_i\left(\{x\in X_{1}:x_{1}<x_{K}+D_1\}\right)\\
\geq \mu_i\left(\{x\in X_{1}:\max\{x_1-D_1,x_2,...,x_{K-1}\}<x_{{K}}\}\right).\end{array}\]
Replacing (\ref{Eq:reverse1}) and (\ref{Eq:equaldist1}) in the first and last expressions of the inequality above yields,
\[\begin{array}{l}\mu_i\left(\{x\in X_{K-1}:x_{K-1}<x_K+D_2\}\right)\\
\geq \mu_i\left(\{x\in X_{1}:\max\{x_1-D_1,x_2,...,x_{K-2},x_K\}<x_{{K-1}}\}\right).\end{array}\]
Now,
\[\mu_i(G\cap X_{K-1})=\mu_i(X_{K-1})-\mu_i\left(\{x\in X_{K-1}:x_{K-1}<x_K+D_2\}\right),\]
and by monotonicity of measures with respect to set inclusion,
\[\begin{array}{rl}\mu_i(G\cap X_{1})&=\mu_i\left(\{x\in X_{1}:\max\{x_1-D_1,x_2,...,x_{K-2},x_K+D_2\}<x_{K-1}\}\right)
\\
&\leq\mu_i\left(\{x\in X_{1}:\max\{x_1-D_1,x_2,...,x_{K-2},x_K\}<x_{K-1}\}\right).\end{array}\]
Thus,
\[\mu_i(G\cap X_{1})+\mu_i(G\cap X_{K-1})\leq \mu_i(X_{K-1})=1/K.\]
Thus, by (\ref{Eq:g-expression1}),
\[B^{u_i,\mu_i}_{ia_{K-1}}(\sigma_{-i})=\mu_i(G)=\mu_i(G\cap X_{1})+\mu_i(G\cap X_{K-1})\leq1/K.\]
\end{proof}

\begin{proof}[Proof of Theorem~\ref{Thm:Paradox2}]Let $\Gamma(u)$ be the game in Table~\ref{Tab:Upsilon1}. By Lemma~\ref{Lm:EUpsilon1}, there is $\sigma^*\in \Nash(\Gamma(u))$ that belongs to the closure of $\{\gamma:(u,\gamma)\in\M\}$ in which each agent $j\neq 1$ plays the strictly dominant action and $\sigma_1^*(a_1)<1/K<\sigma_1^*(a_2)=\dots=\sigma_1^*(a_{K-1})<\sigma_1^*(a_{K})$. Thus, there is $\varepsilon>0$ for which for each $\sigma\in\Delta$ such that $||\sigma-\sigma^*||<\varepsilon$, $\sigma_1(a_1)<1/K<\sigma_1(a_2)$. Let $\sigma$ be such that $(u,\sigma)\in\M$ and $||\sigma-\sigma^*||<\varepsilon$. Since $\sigma_1(a_1)<1/K<\sigma_1(a_2)$ and $(u,\sigma)\in\M$, we have that $U_{u_i}(\sigma_{-i},a_1)<U_{u_i}(\sigma_{-i},a_2)$. Thus, $\sigma_{-i}(a^*_{-i},a_1)<1$, for otherwise $U_{u_i}(\sigma_{-i},a_1)=U_{u_i}(\sigma_{-i},a_2)$. Thus, $U_{u_i}(\sigma_{-i},a_2)-U_{u_i}(\sigma_{-i},a_1)=(1-\sigma_{-i}(a^*_{-i},a_1))>0$ and $U_{u_i}(\sigma_{-i},a_K)-U_{u_i}(\sigma_{-i},a_{K-1})=2(1-\sigma_{-i}(a^*_{-i},a_1))$. Thus,  $U_{u_i}(\sigma_{-i},a_K)-U_{u_i}(\sigma_{-i},a_{K-1})>U_{u_i}(\sigma_{-i},a_2)-U_{u_i}(\sigma_{-i},a_1)$. Thus, there is no $\mu\in\Mcal$ such that $\sigma\in (\Gamma(u),\mu)$, for otherwise by Proposition~\ref{Pro:k3-or-permut--nonmonot1}, $\sigma_1(a_2)\leq 1/K$. Thus,
$\{\sigma:||\sigma-\sigma^*||<\varepsilon\}\cap \{\sigma:\exists\mu\in\Mcal,\textrm{ s.t. }\sigma\in \BNE(\Gamma(u),\mu)\}=\emptyset$.
\end{proof}

\bibliography{ref-BC}

\end{document}